\documentclass[a4paper,UKenglish,cleveref, autoref, thm-restate]{lipics-v2021}
\usepackage{amsmath, amssymb, amsthm, mathrsfs,amsfonts}
\usepackage{bbm}
\usepackage{capt-of}
\usepackage{hyperref}
\usepackage{mathtools}
\usepackage{tikz}
\usepackage{dsfont}
\usetikzlibrary{arrows.meta}
\usetikzlibrary{graphs}

\def\dimt {dim_{TH}}
\def\t #1{\text{#1}}
\def\dima {dim_{\mathcal{A}}}
\def\t #1{\text{#1}}
\def \A{\mathcal{A}}

\def\bg #1{\{G_i=(V,E_i)\}_{i\in[#1]}}
\def\bgp #1{\{G_i'=(V,E_i')\}_{i\in[#1]}}

\def\bgec #1{E(\overline{G_1})\cup E(\overline{G_1}) \cup \cdots \cup E(\overline{G_t})}

\title{Parameterized algorithms for representing graphs as intersections of threshold graphs} 

\author{L. Sunil Chandran}{Indian Institute of Science, India}{sunil@iisc.ac.in}{https://orcid.org/0000-0002-1825-0097}{}

\author{Suraj Kumar Sahoo}{Indian Institute of Science, India}{surajks@iisc.ac.in}{[orcid]}{}

\authorrunning{L. S. Chandran and S. K. Sahoo} 



\acknowledgements{ChatGPT-6 Astra was used to generate Figure 1. }

\nolinenumbers 
\DOIPrefix{}
\hideLIPIcs
\begin{document}

\maketitle

\begin{abstract}
A graph is called a threshold graph if there is an assignment of real valued weights to its vertices such that two vertices have an edge between them if and only if the sum of the weights of the two vertices is above a certain `threshold value'. We will denote the class of all threshold graphs as $TH$.

For a graph $G=(V,E)$, if for some integer $t\geq 0$, there are $t$ many threshold graphs $\{G_i=(V,E_i)\}_{i\in [t]}$ such that we have $E=E_1\cap E_2\cap \cdots \cap E_t$, then we say that the collection of threshold graphs $\{G_i=(V,E_i)\}_{i\in [t]}$ forms a $TH$-representation of $G$ of size $t$. \emph{The intersection dimension of $G$ with respect to threshold graphs} (denoted as $dim_{TH}(G)$) is the smallest non-negative integer $t$ such that there is a $TH$-representation of $G$ of size $t$. 

It is known that for a graph $G$ and integer $t>0$, the problem of deciding whether $dim_{TH}(G)\leq t$ is NP-complete even for $t=3$. In this paper, we study this problem in the framework of parameterized algorithms. We consider the vertex cover number and the cluster vertex deletion number as our parameters. We give a FPT algorithm for determining an optimal $TH$-representation of a graph when parameterized by its vertex cover number. Our main result is an efficient parameterized approximation scheme for this problem when parameterized by the cluster vertex deletion number of $G$.
\end{abstract}
\section{Introduction}
In this paper all graphs are finite, simple and undirected. We will denote a graph $G$ defined on the vertex set $V$ and edge set $E$ as $G=(V,E).$ 

The study of dimensional properties of graphs was initiated by Cozzens and Roberts~\cite{cozzens1989dimensional}, who defined it to be a property $P$ such that any graph can be represented as the intersection of the edge sets of graphs (defined on the same set of vertices) that have the property $P$. The study of dimensional properties of graphs was continued by Kratochvil and Tuza~\cite{Kratochvil}, who showed that these are exactly the classes of graphs that contain all complete graphs and all graphs obtainable by removing one edge of a complete graph. Using this observation, one can define the intersection dimension of graphs as follows.
\begin{definition}[Intersection Dimension of a Graph \cite{Kratochvil}]\label{defdim}
    Let $G=(V,E)$ be a graph and $\mathcal{A}$ be a class of graphs that contains $K_n$ and $K_n\setminus e$ $\forall n\in \mathbb{N}$. \emph{The intersection dimension of $G$ with respect to $\mathcal{A}$} (denoted as $dim_{\mathcal{A}}(G)$) is defined as follows:
    $$dim_{\mathcal{A}}(G)=\min\left\{t\colon  \exists\t{ a collection of graphs }\bg{t}\subset \mathcal{A} \t{ such that }E=\bigcap_{i=1}^tE_i  \right\}.$$
\end{definition}

If $G$ is a complete graph, we take $\dima(G)=0$ by convention. For a graph $G$, we will study $\dima(G)$ when $\A$ is the class of \emph{threshold graphs}, denoted by $TH$. The class of threshold graphs was introduced by Chv\'atal and Hammer~\cite{chvatal1973set} and has applications in linear programming and set packing (see \cite{mahadev1995threshold} for a survey on threshold graphs). Several bounds on this parameter have been established, in terms of its treewidth, pathwidth, vertex cover number, maximum degree, and degeneracy~\cite{chacko2021representing,francis2025some}. Bounds for some special classes of graphs like planar graphs, chordal graphs, AT-free graphs, and co-comparability graphs have also been derived~\cite{chacko2021representing,francis2025some}. 

The problem of determining the intersection dimension of a graph with respect to threshold graphs is computationally difficult, in general. Yannakakis~\cite{Yannakakis} showed that for any integer $k\geq 3$, the problem of determining whether a graph $G$ has $\dimt(G)\leq k$ is NP-complete. However, there is a polynomial time algorithm that can determine whether $\dimt(G)\leq 2$~\cite{Ma,raschle1995recognition,sterbini1998n3}. The computational intractablity of this problem opens way for approaching this problem from the perspective of parameterized complexity. The study of the intersecion dimension of graphs from the view of parameterized algorithms remains quite unexplored.  
 Note that since it is NP-complete to determine if a graph $G$ has $\dimt(G)\leq 3$, parameterization with respect to $\dimt(G)$ is para-NP-hard. Therefore, we will use other parameters like vertex cover number and cluster vertex deletion number to obtain parameterized algorithms. To the best of our knowledge, our results are the first of their kind.

A cluster vertex deletion set is a set $X$ for which $G - X$ is a cluster graph (a cluster graph is a disjoint union of complete graphs). The cardinality of a smallest cluster vertex deletion set of a graph $G$ is called the cluster vertex deletion number of $G$, denoted as $cvd(G)$. A vertex cover is a set $X$ such that $G-X$ is a graph with no edges. The cardinality of a smallest vertex cover is called the vertex cover number of $G$, denoted as $\nu(G)$. Fix integers $k,t\geq 0$. It was shown in \cite{fellows2012well} that the graph classes $\mathcal{G}^k_{vc}=\{G\colon \nu(G)\leq k\}$ and $\mathcal{G}^k_{cvd}=\{G\colon cvd(G)\leq k\}$  are well-quasi ordered by induced subgraphs \footnote{The exact results in the paper are on vertex cover number and \textit{cograph deletion number}. The cograph deletion number of a graph $G$ is the minimum number of vertices that need to be deleted from $G$ so that the rest of the graph is a cograph (a cograph is a graph with no induced $P_4$). But cluster vertex deletion number is at least the cograph deletion number as every cluster is also a cograph. Therefore, the result translates to cluster vertex deletion number as well.}. Let   $\mathcal{H}$ be a hereditary class of graphs. Let $\mathcal{H}^k_{vc}=\mathcal{G}^k_{vc}\cap \mathcal{H}$ and $\mathcal{H}^k_{cvd}=\mathcal{G}^k_{cvd}\cap \mathcal{H}$. For $\mathcal{H}^k_{vc}$, there is a finite set of minimal forbidden graphs in $\mathcal{G}^k_{vc}$ characterizing its membership. This is because the forbidden graphs for $\mathcal{H}^k_{vc}$ that belong to $\mathcal{G}^k_{vc}$ form an antichain in $\mathcal{G}^k_{vc}$ and no such antichain can be infinite because of well quasi orderedness. Similarly, for $\mathcal{H}^k_{cvd}$, there is a finite set of minimal forbidden graphs in $\mathcal{G}^k_{cvd}$ characterizing its membership. The class of graphs $G$ with $\dimt(G)\leq t$ is hereditary. This demonstrates the existence of non-uniform algorithms to recognize graphs that are representable as the intersection of $t$ many threshold graphs given that these graphs belong to $\mathcal{G}^k_{vc}$ or $\mathcal{G}^k_{cvd}$. The aforementioned algorithms can be made uniform by bounding the possible order of the minimal forbidden graphs and exhaustively searching for these forbidden graphs. We devise algorithms that, in spirit, do exactly this. 
\subsection{Our Results}
Suppose $G$ is a graph on $n$ vertices.
\begin{enumerate}
    \item  We prove that the problem of computing a $TH$-representation of $G$ of minimum size is fixed parameter tractable when parameterized by the vertex cover of $G$.
    
    \item  We prove that the problem of computing a $TH$-representation of $G$ of minimum size is fixed parameter tractable when parameterized by both the cluster vertex deletion number of $G$ and $\dimt(G)$. We also provide an efficient parameterized approximation scheme for computing a $TH$-representation of any graph $G$ when parameterized by solely its cluster vertex deletion number.
    
\end{enumerate}
\section{Preliminaries}\label{prelim}
All the graphs in this paper are finite and simple. For a graph $G$, we denote by $\overline{G}$ the complement graph of $G$, i.e., the graph that has the same set of vertices $V(G)$ but with two distinct vertices adjacent if and only if they are not adjacent in $G$. Given graphs $G$ and $H$, we say that graph $H$ is an \textit{induced subgraph} of $G$ if $V(H)\subseteq V(G)$ and $E(H)=\{uv\in E(G)\colon u,v\in V(H) \}$. Given a graph $G$ and a subset of vertices $A\subseteq V(G)$, \textit{the graph induced on $A$ in $G$} is the graph with $A$ as the set of vertices and the set $\{uv\in E(G)\colon u,v\in A \}$ as the set of edges. Given graphs $G$ and $H$, and a subset of vertices $A\subseteq V(G)$, we say that \textit{$A$ induces the graph $H$ in $G$} if the graph induced on $A$ in $G$ is isomorphic to $H$. For a positive integer $x$, we will use the notation $[x]$ to indicate the set $\{1,2,\ldots, x\}$. We will denote by $P_n$, $C_n$, and $K_n$ the path, cycle, and complete graph on $n$ vertices, respectively. 

A collection of threshold graphs $\bg{t}$ such that $E_1\cap E_2\cdots \cap E_t=E$ will be called a \emph{$TH$-representation of $G$ of size $t$} or simply a \emph{$TH$-representation of $G$} when the size $t$ is clear from the context. 

Threshold graphs have many definitions. The following definition will be used most frequently in this paper.
\begin{definition}\label{defth}
    Let $s\in \mathbb{R}_{\geq 0}$ and let $V$ be a finite set. Let $w:V\rightarrow \mathbb{R}_{\geq 0}$ be a function. We define $\Gamma(s,w)$ to be the graph such that the vertex set of $\Gamma(s,w)$ is $V$ and for distinct vertices $u,v\in V$, $uv$ is an edge in $\Gamma(s,w)$ if and only if $w(u) + w(v) \geq  s$. A graph $G=(V,E)$ is said to be a threshold graph if for some $s\in \mathbb{R}_{\geq 0}$ and $w:V\rightarrow \mathbb{R}_{\geq 0}$, $G$ is isomorphic to $\Gamma(s,w)$. In such a case, $\Gamma(s,w)$ is said to be a \emph{threshold representation of $G$}. For a threshold representation $\Gamma(s,w)$ of $G$, $s$ is called its threshold value, and $w$ is called its weight function. Note that we will always implicitly assume that the vertex set of $\Gamma(s,w)$ is the same as the domain of $w$. The class of all threshold graphs is denoted by $TH$. 
\end{definition}
\begin{observation}\label{obsthresholdvaluesmall}
    Let $s\in \mathbb{R}_{\geq 0}$ and let $V$ be a finite set. Let $w:V\rightarrow \mathbb{R}_{\geq 0}$ be a function. Let $w':V\rightarrow \mathbb{R}_{\geq 0}$ be defined as $w'(v)=\min\{w(v),s\}$. Then $\Gamma(s,w)\cong \Gamma (s,w')$.
\end{observation}
Due to Observation~\ref{obsthresholdvaluesmall}, throughout this paper we will assume that the weight function of a threshold representation takes values smaller than or equal to its threshold value.

The following alternate definition will also be used.
\begin{definition}\label{orderdef}
    A threshold graph is a graph that can be constructed from a one-vertex graph iteratively by either adding a single isolated vertex to the graph or adding a single vertex that is connected to all other previously added vertices.
\end{definition}

  We also have the following characterizing theorem for threshold graphs.
\begin{theorem}[\cite{chvatal1973set}]\label{forbth} $TH$ is precisely the class of $(2K_2,C_4,P_4)$-free graphs.
\end{theorem}

Suppose $G=(V,E)$ is a graph. A related problem in this context is the \textit{edge covering} problem that asks for the minimum number of threshold graphs $\bg{t}$ such that $E=\bigcup_{i\in [t]}E_i$. This number is \emph{the threshold cover number of $G$}, denoted as $cov_{TH}(G)$. Observe that computing the intersection dimension of $G$ with respect to threshold graphs is equivalent to computing the threshold cover number of $\overline{G}$ (the complement graph of any threshold graph is a threshold graph and therefore, if there exist threshold graphs $\bgp{t}$ with $E=\bigcap_{i\in [t]}E_i'$, then for $i\in [t]$, we can assign $G_i=\overline{G_i'}$ and get  $\overline{E}=\bigcup_{i\in [t]}\overline{E_i'}$).

The threshold cover number of various graphs was studied by Cozzens and Leibowitz~\cite{doi:10.1137/0605055} and many bounds were provided. We will be using the following result in our proofs.
\begin{theorem}[\cite{doi:10.1137/0605055}]\label{dimtgx}
    For $k\geq 2$, suppose $K(m_1, m_2,\ldots m_{k-1},m_k)$ be the complete $k$-partite graph with $m_1\leq  m_2\leq \cdots m_{k-1}\leq m_k$. Then the threshold cover number of $K(m_1, m_2,\ldots m_{k-1},m_k)$ is $m_{k-1}$. 
\end{theorem}

 Note that there are at most $2^{O(n\log n )}$ threshold graphs on $n$ vertices and the set of all threshold graphs can be enumerated in $2^{O(n\log n )}n^{O(1)}$ time. This is because threshold graphs are a subclass of interval graphs and there are at most $2^{O(n\log n )}$ interval graphs that can be enumerated in $2^{O(n\log n )}n^{O(1)}$ time (see the discussion in Section 2.1 of \cite{Para}). Since the class $TH$ is recognizable in time polynomial in $n$ for any graph $G$ on $n$ vertices, we can simply enumerate all graphs containing $G$  and determine if they belong to $TH$ in time $2^{O(n\log n )}n^{O(1)})$ and we can pick $t$ many graphs from this collection that form a $TH$-representation of $G$. This step requires us to guess a $t$ and enumerate all possible $\binom{2^{O(n\log n )}}{t}\leq2^{O(nt\log n )}n^{O(1)}$ choices and check if intersection of the edge sets of these graphs is equal to the edge set of $G$. 

\section{An FPT algorithm with respect to vertex cover}
Suppose $G=(V,E)$ is a graph and suppose $X$ is a vertex cover of $G$. For every subset $A\subseteq X$, we define the set $S_A$ as follows: for $u,v\in V\setminus X, $ if $N(u)=N(v)=A$ we put them in the set $S_A$. 
\begin{definition}
     We define $G_{vc}$ to be graph obtained by deleting all but 2 vertices from $S_A$ for each $A\subseteq X$ with $|S_A|\geq 2$. We will denote the set so obtained from $S_A$ by $S'_A$.  
\end{definition}

Since there are only $2^{|X|}$ many subsets of $X$, the total number of vertices in $G_{vc}$ is at most $ 2^{|X|+1}+|X|$. We will show that for any graph $G$, $\dimt(G)=\dimt(G_{vc})$ and that we can obtain an optimal $TH$-representation of $G$ from an optimal $TH$-representation $G_{vc}$ in polynomial time.

\begin{lemma}\label{w}
    Suppose $\dimt(G_{vc})\leq t'$. Then there is a $TH$-representation $\left\{\Gamma(s_i,w_i)\right\}_{i\in [t']}$ of $G_{vc}$ such that $\forall j\in [t']$, $\forall A\subseteq X$ and $\forall u,v\in S_A'$, the weight function $w_j$ of $\Gamma(s_j,w_j)$ satisfies $w_j(u)=w_j(v)$. 
\end{lemma}
\begin{proof}
    Let $\left\{\Gamma(s_i,w_i')\right\}_{i\in [t']}$ be a $TH$-representation of $G_{vc}$. For $j\in [t']$, $\forall A\subseteq X$ and $\forall u\in S_A'$ define $w_{j}(u)\colon=\min_{v\in S_A'}w_j'(v)$ and for $x\in X$, define $w_{j}(x)\colon=w_j'(x)$.  
    It is clear that the graphs $\left\{\Gamma(s_i,w_i)\right\}_{i\in [t']}$ satisfy the requirements of the theorem. We claim that it is a $TH$-representation of $G_{vc}$. 
     
    For $j\in [t']$, it is clear that $\Gamma(s_j,w_j)$ is a subgraph of $\Gamma(s_j,w_j')$ because the weights of vertices in $S_A'$ are decreased and no new edges could have formed in such a process. Also $\Gamma(s_j,w_j)$ is a supergraph of $G_{vc}$ because the neighborhood in $ \Gamma(s_j,w_j')$ of any vertex $v$ in $S_A'$ contains $N_{G_{vc}}(S_A')$, including the vertex with minimum weight.  
\end{proof}
\begin{lemma}\label{dteq}
   Let $G=(V,E)$ be a graph. From any $TH$-representation of $G_{vc}$, we can construct a $TH$-representation of $G$ of the same size. In particular, we have $\dimt(G)=\dimt(G_{vc})$.
\end{lemma}
\begin{proof}
    Let $\dimt(G_{vc})\leq t'$. Consider a $TH$-representation $\left\{\Gamma(s_i,w_i')\right\}_{i\in [t']}$ of $G_{vc}$. By Lemma~\ref{w}, we can assume that $\forall j\in [t']$, $\forall A\subseteq X$ and $\forall u,v\in S_A'$ we have $w_j'(u)=w_j'(v)$. 
    
    For $A\subseteq X$ that satisfies $S_{A}'\neq \emptyset$, we will denote by $u_A$ an arbitrarily picked, but fixed, element of $S_A'$.
    For $j\in [t']$, define the weight function $w_j$ as follows: for $A\subseteq X$ and for $x\in S_A$, $w_j(x)\colon =w_j'(u_A)$ and for $x\in X$, $w_j(x)\colon =w_j'(x)$. We claim that $\left\{\Gamma(s_i,w_i)\right\}_{i\in [t']}$ is a $TH$-representation of $G$. 
   
      Let $x\in V\setminus X$ and $y\in X$. Suppose $xy\in E$. Let $x\in S_{A}$ for some $A\subseteq X$. Then $u_Ay$ is an edge in $G_{vc}$ as $N_G(x)=N_G(u_A)=N_{G_{vc}}(u_A)$. Thus, $\forall j\in [t']$, $w_j(x)+w_j(y)=w'_j(u_A)+w_j'(y)\geq s_j$. Now, suppose $xy\notin E$. If for some $A\subseteq X$ we have $x\in S_A$, then $\exists i\in [t']$ such that $w_i'(u_A)+w_i'(y)<s_i$ since $\left\{\Gamma(s_i,w_i')\right\}_{i\in [t']}$ is a $TH$-representation of $G_{vc}$. Consequently, $w_i(x)+w_i(y)=w_i'(u_A)+w_i'(y)<s_i$. 
   
    Finally, suppose for a pair of (not necessarily distinct) subsets $A,B\subseteq X$ we have $x\in S_A, y\in S_B$ with $x\neq y$. If $A\neq B$, $\exists j\in [t']$ such that $w_j'(u_A)+w_j'(u_B)<s_j$ as $\left\{\Gamma(s_i,w_i')\right\}_{i\in [t']}$ is a $TH$-representation of $G_{vc}$. Consequently, $w_j(x)+w_j(y)=w_j'(u_A)+w_j'(u_B)<s_j$. If $A=B$, then $x,y\in S_A$ and therefore, $|S_A'|\geq 2$. Let $u\in S_A'\setminus \{u_A\}$. Then $\exists j\in [t']$ such that $w_j'(u_A)+w_j'(u)<s_j$. Consequently, $w_j(x)+w_j(y)=w_j'(u_A)+w_j'(u_A)=w_j'(u_A)+w_j'(u)<s_j$, where the second equality is because of Lemma~\ref{w}. 

    Hence, $\left\{\Gamma(s_i,w_i)\right\}_{i\in [t']}$ is a $TH$-representation of $G$. Clearly this constructive procedure is also polynomial time. 
\end{proof}
It was shown in~\cite{francis2025some} that for any graph $G$ with minimum vertex cover number $k\geq 1$, we have $\dimt(G)\leq k$. By the discussion in Section~\ref{prelim}, the optimal $TH$-representation of $G_{vc}$ can be computed in $2^{O(k^22^k)}n^{O(1)}$ time. From this $TH$-representation, we can obtain an optimal $TH$-representation for $G_{vc}$ that satisfies the statement of Lemma~\ref{w} in polynomial time. This $TH$-representation can be extended to an optimal $TH$-representation of $G$ in polynomial time due to Lemma~\ref{dteq}. Hence, we have the following theorem. 

\begin{theorem}
    For a graph $G$ on $n$ vertices with vertex cover number $k$, there is an algorithm that constructs an optimal  $TH$-representation of $G$ in time $2^{O(k^22^k)}n^{O(1)}$.
\end{theorem}

\section{An efficient parameterized approximation scheme with respect to cluster vertex deletion number}

Suppose $G$ is a graph and $X$ is a cluster vertex deletion set of $G$. Let $\mathcal{C}$ be the family of connected components of $G-X$. Since $X$ is a cluster vertex deletion set of $G$, $\mathcal{C}$ is a family of cliques. We will say that cliques $C,C'\in \mathcal{C}$ are equivalent, denoted by $C\sim C'$, if there is a bijective function $f\colon V(C)\rightarrow V(C')$ such that $\forall u\in V(C), N(u)\cap X=N(f(u))\cap X$. It is easy to see that $\sim$ defines an equivalence relation on $\mathcal{C}$. We will denote the equivalence class containing $C\in \mathcal{C}$ by $\langle C \rangle $. 

In this section, the general proof strategy is to produce a kernel for $G$ by bounding the number of vertices in each clique and the number of cliques in each equivalence class under $\sim$. In \cite{bruhn2016structural}, the authors show that the problem of finding the box-representation of a graph is FPT when parameterized by the cluster vertex deletion number. We will abstract out the key combinatorial idea in their proof and use it in our context to bound the number of cliques in each equivalence class under $\sim$. To bound the number of vertices in each clique, we will analyze the structural properties of the graph class $TH$. 

We first make some observations. 
\begin{observation}\label{atmonevc}
     Suppose $\Gamma(s,w)$ is a threshold graph and $C$ is a clique of $\Gamma(s,w)$. Then, there can be at most one vertex $v\in V(C)$ such that $w(v)<s/2$.
\end{observation}
\begin{proof}[Proof of Observation]
   Otherwise two vertices $v$ and $v'$ are such that $w(v)<s/2$ and $w(v')<s/2$. Consequently, $w(v)+w(v')<s$ which is a contradiction to the fact that $v$ and $v'$ belong to the same clique.
\end{proof}
\begin{observation}\label{4graph}
    Suppose $G$ is a graph on 4 vertices $a,b,c,d$ such that $ab$ and $cd$ are edges. If there is a matching $M$ of size 2 in $\overline{G}$, then $\dimt(G)>1$.
\end{observation}
\begin{proof}[Proof of Observation]
    We can always find a $2K_2, C_4$ or $P_4$ in such a graph. By Theorem~\ref{forbth}, $G$ cannot be a threshold graph.
\end{proof}
Co-bipartite graphs are the graphs $G$ with the property that $\overline{G}$ is bipartite.  
\begin{theorem}\label{thpoly}
     Suppose $G=(A,B,E)$ is a co-bipartite graph. If $k$ is the vertex cover number of $\overline{G}$, then $\dimt(G)= k$. Furthermore, we can find a $TH$-representation of $G$ of size $k$ in polynomial time.
\end{theorem}
\begin{proof}
     Let $M$ be a maximum matching of $\overline{G}$. If $uv,u'v'\in M$ are distinct edges in $\overline{G}$, then by Observation~\ref{4graph}, no threshold graph that contains $G$ can contain the graph induced on $uvu'v'$ (any threshold graph containing $G$ is also co-bipartite) as an induced subgraph. Therefore, each threshold graph containing $G$ can contain at most one element of $M$ as a non-edge. Therefore, $\dimt(G)\geq k$. 
     
     We will show that $\dimt(G)\leq k$ by covering $\overline{G}$ by $k$ many threshold graphs. Let $X$ be a minimum vertex cover of $\overline{G}$. By a theorem of K\H{o}nig, for a bipartite graph, the matching number of a graph is equal to its vertex cover number. 
    Therefore, $|X|=k$. For every vertex $v\in X$, define the threshold graph $G_v'$ to be the graph induced by $N_{\overline{G}}[v]$ 
    in $\overline{G}$. Let $G_v$ be the threshold graph obtained from $G_v'$ by adding all the other vertices in $G$ as isolated vertices. Since $X$ is a vertex cover of $\overline{G}$, the edges of $\overline{G}$ are all covered by the edges of the graphs $\{G_v\}_{{v\in X}}$. 
    Therefore, $\dimt(G)\leq k$. The time taken to compute $\dimt(G)$ and a $TH$-representation of $G$ is the time taken to find the minimum vertex cover of $\overline{G}$. Since $\overline{G}$ is a bipartite graph, this can be done in polynomial time.      
\end{proof}

A threshold graph can have multiple threshold representations. The following theorem allows us to use a single representation for ease of exposition. 

\begin{lemma}\label{integth}
    Let $k$ be a positive integer. Suppose $G=(V,E)$ is a threshold graph and $V'$ is a subset of $V$ of order $k$. Then there is a threshold representation $\Gamma(2k+3,w)$ of $G$ such that $\forall v\in V-V'$, $w(v)\notin [2k+3]$ and $\forall v\in V', w(v)\in [2k+3]\setminus\{k+1,k+2\}$.
\end{lemma}

The proof of this Lemma has been moved to the appendix.

From this point on, we will take $k$ as the cardinality of a cluster vertex deletion set $X$ of a graph $G$. Due to Lemma~\ref{integth}, we can assume that any threshold graph that contains $G$ as a subgraph must be of the form $\Gamma(2k+3,w)$, where $w$ satisfies the conditions of Lemma~\ref{integth}.
\begin{definition}\label{config}
    Suppose $G=(V,E)$ is a graph and suppose $w:V\rightarrow \mathbb{R}$ is a weight function. We define an indicator function $c(w,\cdot)$ as $c(w,v)=0$ if $w(v)<\frac{1}{2}(2k+3)$ and $c(w,v)=1$ otherwise. We define the function $g(w,\cdot)$ as $g(w,v)=w(v)$ if $c(w,v)=0$ and $g(w,v)=2k+3-w(v)$ if $c(w,v)=1$. Note that we will always implicitly assume that $c(w,\cdot)$ and $g(w,\cdot)$ are defined on the same domain as $w$.  
\end{definition}
\begin{remark}
    When $\Gamma(2k+3,w)$ is a representation of a threshold graph, the functions $c(w,\cdot)$ and $g(w,\cdot)$ encode important information about the behavior of the vertices of $G$. In fact, if for a vertex $v$, we know both the values $c(w,v)$ and $g(w,v)$, then we can easily deduce the value of $w(v)$. We will be crucially using these two functions in our analysis. The intuition behind using these functions comes from interval graphs. Every threshold graph with a representation $\Gamma(s,w)$ has a representation as interval graphs as follows: (i) if a vertex $v$ has $w(v)<s/2$ in $\Gamma(s,w)$, then we assign to it the point interval $[w(v),w(v)]$ and (ii) if a vertex $v$ in $G$ has $w(v)\geq s/2$, we assign it the interval $[s-w(v),s]$. Assuming that $w$ takes distinct values, it is easy to check that such a representation is a valid interval representation of $\Gamma(s,w)$ (see Figure~\ref{fig:1}). Note that the left endpoints of these intervals are exactly determined by $g(w,\cdot)$ when $s=2k+3$. Most of our analysis using the functions $g(w,\cdot)$ and $c(w,\cdot)$ is done with Figure~\ref{fig:1} in mind.
\end{remark}
\begin{figure}[!ht]
    \centering

\resizebox{0.75\textwidth}{!}{
\begin{tikzpicture}[
    x=0.9cm,y=0.95cm,
    font=\small,
    line cap=round,
    baselineguide/.style={draw=black!25,line width=0.35pt},
    endpoint/.style={line width=0.9pt},
    projection/.style={densely dashed,line width=0.65pt},
    rowlabel/.style={anchor=east,inner sep=0pt},
    every node/.style={inner sep=2pt}
]
\definecolor{edgeblue}{RGB}{24,96,153}
\definecolor{nonedgered}{RGB}{163,53,57}
\definecolor{intervalink}{RGB}{38,58,72}

\node[font=\normalsize] at (4.25,5.15)
  {$s=10,\qquad w(v)=2,\quad w(u)=4,\quad w(h)=7$};

\draw[baselineguide] (0,4.25)--(10,4.25);
\draw[baselineguide] (0,3.40)--(10,3.40);
\draw[baselineguide] (0,2.35)--(10,2.35);

\draw[projection] (2,1.25)--(2,4.25);
\draw[projection] (4,1.25)--(4,3.40);
\draw[projection,black!40] (3,1.25)--(3,2.35);
\draw[projection,black!40] (10,1.25)--(10,2.35);

\node[rowlabel] at (-0.38,4.25) {$I(v)=\{2\}$};
\node[anchor=west,fill=white] at (6.05,4.25) {$w(v)<s/2$};

\node[rowlabel] at (-0.38,3.40) {$I(u)=\{4\}$};
\node[anchor=west,fill=white] at (6.05,3.40) {$w(u)<s/2$};

\node[rowlabel] at (-0.38,2.35) {$I(h)=[3,10]$};
\draw[line width=2pt] (3,2.35)--(10,2.35);
\draw[endpoint] (3,2.22)--(3,2.48);
\draw[endpoint] (10,2.22)--(10,2.48);
\node[anchor=west,fill=white] at (3,2.73) {$s-w(h)=3$};
\node[anchor=east] at (10,2.73) {$w(h)\ge s/2$};

\draw[-{Stealth[length=4pt]},black!75,line width=0.65pt]
  (0,1.25)--(10.40,1.25);
\foreach \a in {0,2,3,4,5,10}
  {\draw[black!7,line width=0.6pt] (\a,1.18)--(\a,1.32);}
\foreach \a in {0,2,3,4}
  {\node[below] at (\a,1.15) {$\a$};}
\node[below] at (5,1.15) {$s/2=5$};
\node[below] at (10,1.15) {$s=10$};

\node[align=center] at (1.65,0.15)
  {$I(v)\cap I(h)=\varnothing$\\[2pt]
   $2+7<10\quad\Longrightarrow\quad vh\notin E(G)$};
\node[align=center] at (7.30,0.15)
  {$I(u)\cap I(h)=\{4\}$\\[2pt]
   $4+7\ge10\quad\Longrightarrow\quad uh\in E(G)$};

\end{tikzpicture}
}
    \caption{An example of an interval representation of a threshold graph. Filled dots are singleton intervals. All intervals lie on the same real
   line; rows are offset only for readability. Pale lines and dashed
   projections are guides, not additional intervals.}
    \label{fig:1}
\end{figure}
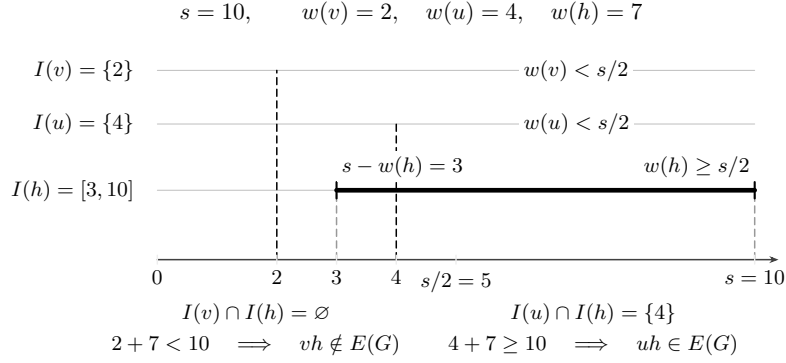

We can rehash the statement of Lemma~\ref{integth} in a more useful way by stating it in terms of the function $g(w,\cdot)$ and the cluster vertex deletion set as follows.

\begin{corollary}\label{integth1}
    Let $k$ be a positive integer. Suppose $G=(V,E)$ is a threshold graph and $X$ is a subset of $G$ of cardinality $k$. Then there is a threshold representation $\Gamma(2k+3,w)$ of $G$ such that the following three conditions hold: \begin{romanenumerate}
        \item $g(w,\cdot)$ takes distinct values.
        \item $\forall v\in X$,  $g(w,v)\in [k]$.
        \item $\forall v\in V-X$, $\exists j\in [k+1]$ such that $j-1< g(w,v)< j$.
    \end{romanenumerate}   
\end{corollary}

We will say that two vertices $u$ and $v$ of $G$ have the \emph{same configuration in $\Gamma(2k+3,w)$} if $\lceil g(w,u)\rceil=\lceil g(w,v)\rceil$ and $c(w,u)=c(w,v)$.
\begin{lemma}\label{gval}
         Suppose $G=(V,E)$ is a graph and $X$ is a cluster vertex deletion set of $G$ of cardinality $k$. Let $\Gamma(2k+3,w)$ be threshold representation of a threshold graph containing $G$ as a subgraph that satisfies conditions (ii) and (iii) of Corollary~\ref{integth1}. Suppose that two vertices $u$ and $v$ of $G-X$ have the same configuration in $\Gamma(2k+3,w)$.  
         Then in the graph  $\Gamma(2k+3,w)$, $u$ and $v$ have exactly the same set of neighbors and non-neighbors in $X$.
     \end{lemma}
     \begin{proof}
         Let $a\in X$ and let $g(w,a)=j$. 
         
         Suppose $au$ is an edge in $\Gamma(2k+3,w)$. We will show that $av$ is an edge in $\Gamma(2k+3,w)$. If $w(u)<\frac{1}{2}(2k+3)$, then we must have $w(a)\geq\frac{1}{2}(2k+3)$ and consequently $w(u)\geq 2k+3-w(a)=g(w,a)=j$. By our choice of the $TH$-representation, we know that we must have $w(u)>j$. Since $u$ and $v$ have the same configuration in $\Gamma(2k+3,w)$, we must have $w(v)<\frac{1}{2}(2k+3)$ and $w(v)>j=2k+3-w(a)$. Therefore, $av$ is an edge in $\Gamma(2k+3,w)$. On the other hand, if $w(u)\geq \frac{1}{2}(2k+3)$ then we have $w(v)\geq\frac{1}{2}(2k+3)$. If we also have $w(a)\geq \frac{1}{2}(2k+3)$ we must have $av$ as an edge in $\Gamma(2k+3,w)$ and we are done. So assume $w(a)<\frac{1}{2}(2k+3)$. Then $g(w,u)=2k+3-w(u)\leq w(a)=g(w,a)=j$. As earlier, we must have strict inequality $g(w,u)=2k+3-w(u)<j$. Since $u$ and $v$ have the same configuration in $\Gamma(2k+3,w)$, we must have $2k+3-w(v)\leq \lceil2k+3-w(v)\rceil =\lceil g(w,v)\rceil=\lceil g(w,u)\rceil\leq j=w(a)$ and so $av$ is an edge in $\Gamma(2k+3,w)$.
         
         Suppose $au$ is not an edge in $\Gamma(2k+3,w)$. We will show that $av$ is also not an edge in $\Gamma(2k+3,w)$. If $w(u)\geq \frac{1}{2}(2k+3)$, we must have $w(a)<\frac{1}{2}(2k+3)$ and consequently $j=g(w,a)=w(a)< 2k+3-w(u)=g(w,u)$. By our choice of the $TH$-representation, we know that we must have $j<2k+3-w(u)$. Since $u$ and $v$ have the same configuration in $\Gamma(2k+3,w)$, we must have $w(v)\geq \frac{1}{2}(2k+3)$ and $w(a)=j<2k+3-w(v)$. Therefore, $av$ is not an edge in $\Gamma(2k+3,w)$. If $w(u)< \frac{1}{2}(2k+3)$ then we have $w(v)< \frac{1}{2}(2k+3)$. If we also have $w(a)< \frac{1}{2}(2k+3)$ then we cannot have $av$ as an edge in $\Gamma(2k+3,w)$ and we are done. So assume $w(a)\geq  \frac{1}{2}(2k+3)$. Then $j=g(w,a)=2k+3-w(a)> w(u)=g(w,u)$. As earlier, we must have the inequality $w(u)<j$. Since $u$ and $v$ have the same configuration in $\Gamma(2k+3,w)$, we must have $w(v)< j=2k+3-w(a)$ and so $av$ is not an edge in $\Gamma(2k+3,w)$.
     \end{proof}
\begin{lemma}\label{struct1}
    Suppose $G$ is a graph and let $\dimt(G)=t$. Suppose $X$ is a cluster vertex deletion set of $G$ of cardinality $k\geq 1$ and is given as input. Let $l$ and $p$ be positive integers such that the largest and second largest cliques in $G-X$ have order $l$ and $p$ respectively. Then we can obtain an induced subgraph $G'$ of $G$ in polynomial time such that $G'$ has order at most $2p\cdot(2^k+p)^p(2^{tp}(k+1)^{tp}+1)+l+k$ and $\dimt(G')=t$. Moreover, we can obtain a $t$ sized $TH$-representation of $G$ from a $t$ sized $TH$-representation of $G'$ in polynomial time.   
\end{lemma}
\begin{proof}
     We will show that if an equivalence class of cliques $\mathcal{G}\subseteq  \mathcal{C}$ contains at least $2^{tp}(k+1)^{tp}+2$ many cliques, then $\forall C\in \mathcal{G}$, we have $\dimt(G)=\dimt(G-C)$. Due to this fact, we can retain at most $2^{tp}(k+1)^{tp}+1$ many cliques in each equivalence class and remove all extra cliques. Because each clique other than the largest one is of order at most $p$, the number of equivalence classes is at most $\sum_{r=1}^{p}\binom{2^k+r-1}{r}
 =\binom{2^k+p}{p}-1\leq (2^k+p)^p$. It follows that if $l=p$, we have a kernel of order $2p\cdot(2^k+p)^p (2^{tp}(k+1)^{tp}+1) +k$. On the other hand, if $p<l$, then the equivalence class containing the largest clique is a singleton class and therefore, we will end up with a kernel of order  $2p\cdot(2^k+p)^p (2^{tp}(k+1)^{tp}+1)+l+k$ as claimed.

     To that end, let $\tilde{C}\in \mathcal{G}$ and $|\mathcal{G}|\geq2^{tp}(k+1)^{tp}+2$. Consider the graph $G'=G-\tilde{C}$. It is enough to show that $\dimt(G')=\dimt(G)$ and that we can obtain an optimal $TH$-representation of $G$ from an optimal $TH$-representation of $G'$ in polynomial time.   
     
     Let $t'=\dimt(G')$. By Corollary~\ref{integth1}, we can assume that we have a $TH$-representation $\left\{\Gamma(2k+3,w_i')\right\}_{i\in [t']}$ of $G'$ such that $\forall i\in [t'],$ $(i)$ $g(w_i',\cdot)$ takes distinct values, $(ii)$ $\forall v\in X$, $g(w_i',v)\in [k]$ and (iii) $\forall v\in V-X$, $\exists j\in [k+1]$ such that $j-1< g(w'_i,v)< j$.  

     $\forall i\in [t'],$ we define a new function $w_i$ using the definition of $w_i'$ as follows:
     \begin{romanenumerate}
           
         \item Suppose $C\in \mathcal{C}\setminus\{\tilde{C}\}$ and suppose $v$ is a vertex of $C$. 
         
         Assign $g(w_i,v):=\max\{g(w_i',u)\colon u \t{ is in }C \t{ and }\lceil g(w_i',u)\rceil= \lceil g(w_i',v)\rceil\}$. 
         
         Let $w_i(v)=\begin{cases}
g(w_i,v), & \text{if } c(w_i',v)=0,\\
2k+3-g(w_i,v), & \text{if } c(w_i',v)=1.
\end{cases}$.
         \item For any other vertex $v$ of $G'$, we assign $w_i(v)=w_i'(v)$.
     \end{romanenumerate}
     \begin{claim}\label{w_i}
     $\{\Gamma(2k+3,w_i)\}_{i\in [t']}$ is a $TH$-representation of $G'$.
\end{claim}

     \begin{claimproof}
         Suppose $\forall i\in [t'], w_i$ satisfies the above requirements. Fix an $i\in [t']$. Suppose for some vertex $v$ in some clique $C$, $w_i(v)\neq w_i'(v)$. By construction, we must have $c(w_i',v)=1$. It is clear that $c(w_i,v)=c(w_i',v)=1$ and $g(w_i,v)=g(w_i,u)=g(w_i',u)$ and therefore, $\lceil g(w_i,v)\rceil=\lceil g(w_i',v)\rceil$. Hence, by arguments similar to Lemma~\ref{gval}, the adjacencies and non-adjacencies of $v$ remain the same with respect to the vertices in $X$. The vertices in $C$ remain adjacent to each other. Observe that for every vertex $v$ in $G'$, $w_i(v)\leq w_i'(v)$. Therefore, no new edges are formed. 
     \end{claimproof}
     \begin{claim}\label{conf}
         If $C\in \mathcal{C}$, then for every pair of vertices $u,v$ in $C$ that satisfies $\lceil g(w_i,v)\rceil=\lceil g(w_i,u)\rceil$, we have $g(w_i,v)= g(w_i,u)$. In particular, if there is a vertex $u$ in $C$ such that $c(w_i,u)=0$, then for every other vertex $v\neq u$ in $C$, we have $g(w_i,v)\leq g(w_i,u)$.
     \end{claim}
     \begin{claimproof}
         Follows directly from construction.
     \end{claimproof}
     We will say that two distinct cliques $C_1,C_2\in \langle \tilde{C}\rangle$ have the \emph{same configuration in $\{ \Gamma(2k+3,w_i)\}_{i\in [t']}$} if there is a bijection $f$ that maps the vertices of $C_1$ to the vertices of $C_2$ such that for every vertex $ v$ in $ C_1$, $\forall i\in [t']$, $v$ and $f(v)$ have the same configuration in $\Gamma(2k+3,w_i)$. 
     
     It is easy to see that the number of cliques that do not have the same configuration in $\{ \Gamma(2k+3,w_i)\}_{i\in [t']}$ is at most $2^{t'p}\times (k+1)^{t'p}$ (for every vertex $v$ and for each $i\in [t']$ there are 2 choices for the value of $c(w_i,v)$ and $k+1$ many choices for the value of $\lceil g(w_i,v)\rceil$).  
     Since we have more cliques than that, there are two cliques $C_1,C_2\in \langle \tilde{C}\rangle $ that have the same configuration in $\{ \Gamma(2k+3,w_i)\}_{i\in [t']}$. 

     Suppose $q$ is the bijective function that maps the vertices of $\tilde{C}$ to the vertices of $C_1$ such that for every vertex $ v$ in $ \tilde{C},$ we have $N(v)\cap X=N(q(v))\cap X$ (such a function exists because $\tilde{C}\sim C_1$). For every $i\in [t']$, we will extend the definition of $w_i$ to include the vertices of $\tilde{C}$ by extending the definitions of $g(w_i,\cdot)$ and $c(w_i,\cdot)$ respectively. For every vertex $v$ in $\tilde{C}$, assign $c(w_i,v)=c(w_i,q(v))$.  
     
     Let $\delta_i'=\min_{w_i(u)+w_i(v)<2k+3}\{2k+3-w_i(u)-w_i(v)\}$ and $\delta_i''= \min_{g(w_i,u)\neq g(w_i,v)}\{|g(w_i,u)-g(w_i,v)|\}$. Let $\delta_i=\frac{1}{2}\min\{\delta_i',\delta_i''\}$. 
    
     Note that for a vertex $ v$ in $ \tilde{C}$, $g(w_i,q(v))\neq g(w_i,f(q(v)))$. For a vertex $ v$ in $ \tilde{C}$, assign $g(w_i,v)$ as \begin{equation}\label{eq}
         g(w_i,v)=\begin{cases}
         g(w_i,q(v))+\delta_i & \t{ if }g(w_i,q(v))<g(w_i,f(q(v)))\\
          g(w_i,q(v))-\delta_i & \t{ if }g(w_i,q(v))>g(w_i,f(q(v)))
          \end{cases}.
     \end{equation}
     Since $\forall i\in [t']$, $g(w_i,\cdot)$ and $c(w_i,\cdot)$ values are known on the vertex set of $G$, the values of $w_i$ is completely determined.
     Note that due to our choice of $\delta_i,$ for every vertex $v$ in $\tilde{C}$, we have $\lceil g(w_i,v)\rceil=\lceil g(w_i,q(v))\rceil=\lceil g(w_i,f(q(v)))\rceil$. It is also clear that we have $w_i(v)=w_i(q(v))\pm \delta_i$. Therefore, $\tilde{C}$, $C_1$ and $C_2$ are in the same configuration in $\{ \Gamma(2k+3,w_i)\}_{i\in [t']}$.

     \begin{claim}\label{claimt}
         $\{\Gamma(2k+3,w_i)\}_{i\in [t']}$ is a $TH$-representation for $G$.
     \end{claim}
     \begin{claimproof}
         The only thing that needs to be checked are the adjacencies and non-adjacencies of the vertices in $\tilde{C}$. As discussed above, $\forall i\in [t']$, we know that $\tilde{C}$, $C_1$ and $C_2$ are in the same configuration in $\{ \Gamma(2k+3,w_i)\}$.  
         Therefore, by Lemma~\ref{gval}, the adjacencies and non-adjacencies of $\tilde{C}$ with respect to the vertices of $X$ are resolved. It is also easy to observe that the vertices of the clique $\tilde{C}$ are adjacent to each other in each graph $\Gamma(2k+3,w_i)$. Therefore, we focus on the adjacencies and non-adjacencies of vertices in $\tilde{C}$ with the vertices of all the other cliques in $\mathcal{C}$.

         Let $v$ be a vertex of $\tilde{C}$. We consider two cases.
         \begin{enumerate}
             \item Suppose that $u$ is a vertex from $G-X$ not in $\tilde{C}$ or $C_1$. We know that $\exists i\in [t']$ such that $w_i(u)+w_i(q(v))<2k+3$. Clearly, it is true that $w_i(u)+w_i(q(v))-\delta_i <2k+3$. By our choice of $\delta_i$, we also have $w_i(u)+w_i(q(v))+\delta_i<w_i(u)+w_i(q(v))+ 2k+3-w_i(u)-w_i(q(v)) =2k+3$. Since $w_i(v)=w_i(q(v))\pm \delta_i$, we have $w_i(u)+w_i(v)<2k+3$ and we are done. 
             \item
             Suppose that $u$ is a vertex from $C_1$. We know that $\exists i\in [t']$ such that $w_i(u)+w_i(f(q(v)))<2k+3$.
         \begin{romanenumerate}
             \item If $c(w_i,u)=c(w_i,f(q(v)))=0$, then $c(w_i,v)=0$ because $v,q(v)$ and $f(q(v))$ have the same configuration in $\Gamma(2k+3,w_i)$. Consequently, $w_i(u)+w_i(v)<\frac{1}{2}(2k+3)+\frac{1}{2}(2k+3)=2k+3$ and therefore, $uv$ is not an edge in $\Gamma(2k+3,w_i)$. 
             \item Now assume $c(w_i,u)=0$ and $c(w_i,f(q(v)))=1$. This means that \begin{equation}\label{eq2}
                 c(w_i,f(u))=0 \t{ and }c(w_i,v)=1.
             \end{equation} because $u$ and $f(u)$ have the same configuration in $\Gamma(2k+3,w_i)$ and 
             $v,q(v)$ and $f(q(v))$ have the same configuration in $\Gamma(2k+3,w_i)$. Hence, by Claim~\ref{conf} and expression~(\ref{eq2}), we have $g(w_i,f(q(v)))\leq g(w_i,f(u))$. 
             
            Since $c(w_i,u)=0$ and $c(w_i,f(q(v)))=1$, we have \begin{equation}\label{eq3}
                 g(w_i,u)=w_i(u)<2k+3-w_i(f(q(v)))=g(w_i,f(q(v)))\leq  g(w_i,f(u)).
             \end{equation} 
             Since  $u$ and $f(u)$ have the same configuration in $\Gamma(2k+3,w_i)$, it follows from expression~(\ref{eq3}) that 
             \begin{equation}\label{eq4}
                 \lceil g(w_i,u)\rceil =\lceil g(w_i,f(u))\rceil =\lceil g(w_i,f(q(v)))\rceil=\lceil g(w_i,v)\rceil=\lceil g(w_i,q(v))\rceil.
             \end{equation}
         By Claim~\ref{conf} and expression~(\ref{eq4}), we have \begin{equation}\label{eq5}
             g(w_i,q(v))=g(w_i,u)<g(w_i,f(q(v))).
         \end{equation} Therefore, by Equation~(\ref{eq}), $g(w_i,v)=g(w_i,q(v))+\delta_i $ and consequently, we have \begin{align*}
             w_i(v)&=2k+3-g(w_i,v) & \t{ (by definition of }g(w_i,\cdot) \t{ and expression (2)})\\&=2k+3-g(w_i,q(v))-\delta_i& \t{ (by Equation (1))}\\&=2k+3-g(w_i,u)-\delta_i& \t{ (by Equation (5))}\\&=2k+3-w_i(u)-\delta_i& \t{ (by definition of }g(w_i,\cdot)\t{ and assumption}).
         \end{align*}Therefore, $w_i(u)+w_i(v)=w_i(u)+2k+3-w_i(u)-\delta_i=2k+3-\delta_i<2k+3$.
         \item Finally assume that $c(w_i,u)=1$ and $c(w_i,f(q(v)))=0$. The proof is similar to the proof of case 2 part (ii) and therefore, has been moved to the appendix. 
         \end{romanenumerate}
         \end{enumerate}
     \end{claimproof}
     Therefore, we showed that $t=\dimt(G)\leq t'\leq t$. Hence, $\dimt(G)=\dimt(G-\tilde{C})$. 
     \begin{claim}\label{poly}
         From a $TH$-representation of $G'$ we can obtain a $TH$-representation of $G$ in polynomial time.
     \end{claim}
     The proof of Claim~\ref{poly} has been moved to the appendix.
\end{proof}
\begin{lemma}\label{cliqmax}
    Suppose $G=(V,E)$ is a graph and suppose $X$ is a cluster vertex deletion set of $G$. Let $|X|=k$ and $\dimt(G)= t$. Then there is a subgraph $G'=(V',E')$ of $G$ such that $\dimt(G')=\dimt(G)$ and the largest clique in $G'-X$ has order at most $2^k(t+1)$. Furthermore, it is possible to obtain a $TH$-representation of $G$ from a $TH$-representation of $G'$ in polynomial time. 
\end{lemma}
\begin{proof}
    Let $\left\{\Gamma(2k+3,w_i)\right\}_{i\in [t]}$ be a $TH$-representation of $G$. Suppose $C_{max}$ is the largest clique in $G-X$.
     
     If the number of vertices $C_{max}$ is less than or equal to $ 2^k(t+1)$, then $G=G'$ and we are done. Therefore, suppose $C_{max}$ has atleast $2^k(t+1)+1$. Define an equivalence relation $\sim_{mc}$ on the vertices of $C_{max}$ in the following way: for distinct vertices $u$ and $v$ in $C_{max}$, $u\sim_{mc} v$ if and only if $N_{G}(u)\cap X=N_{G}(v)\cap X$. Let $\mathcal{F}_{mc}$ be the family of all equivalence classes under $\sim_{mc}$ that contain at least $t+1$ vertices.
     \begin{claim}
         $\mathcal{F}_{mc}$ is non-empty.
     \end{claim}
     \begin{claimproof}
         Since there are at most $2^k$ distinct subsets of $X$, there can be at most $2^k$ distinct equivalence classes under $\sim_{mc}$. Because $C_{max}$ contains at least $2^k(t+1)+1$ vertices, by the pigeonhole principle, there is at least one equivalence class that contains at least $t+1$ vertices. 
     \end{claimproof}
       
     Let $G'$ be the graph obtained from $G$ by retaining exactly $t+1$ vertices from $F$ for every $F\in \mathcal{F}_{mc}$. 
     \begin{claim}\label{claimth}
         $\dimt(G')=\dimt(G)$.
     \end{claim}
     \begin{claimproof}
         Fix an equivalence class $F\in \mathcal{F}_{mc}$. Let $G_F=(V_F,E_F)$ obtained from $G$ by retaining exactly $t+1$ vertices from $F$. It is sufficient to prove that $\dimt(G_F)=\dimt(G)$ because we can do this step recursively for every equivalence class in $\mathcal{F}_{mc}$ and the claim follows. 
         
         Suppose $t'=\dimt(G_F)$ and $\left\{\Gamma(2k+3,w_i')\right\}_{i\in [t']}$ is an optimal $TH$-representation for $G_F$. 
          Note that by Observation~\ref{atmonevc}, for a fixed $i\in [t']$ there can be at most one vertex $v\in F\cap V_F$ such that $w_i'(v)<\frac{1}{2}(2k+3)$. Since $|F\cap V_F|= t+1> t\geq t'$, there is a vertex $u_F$ in $F\cap V_F$ such that $\forall i\in [t']$, we have $w_i'(u_F)\geq \frac{1}{2}(2k+3)$.
          
           $\forall i\in [t'],$ define $w_i^*$ on the vertices of $G$ as follows: for $u\in V\setminus F$, assign $w_i^*(u)=w_i'(u)$ and for $u\in F$, assign $w_i^*(u)=w_i'(u_F)$.  It is easy to see that $\{\Gamma(2k+3,w_i^*)\}_{i\in [t']}$ is a $TH$-representation for $G$. Therefore, $t= t'$.
     \end{claimproof}
     Note that the vertices of $C_{max}$ that are retained in $G'$ induces a clique of order at least $|\mathcal{F}_{mc}|(t+1)\geq (t+1)>t$ and therefore, by Theorem~\ref{dimtgx}, it remains the largest clique in $G'-X$. Since there are at most $2^k$ many equivalence classes under $\sim_{mc}$, the largest clique in $G'$ has at most $2^k(t+1)$ vertices. Therefore, $G'$ is the required graph. 
     \end{proof}
\begin{theorem}\label{EPAS}
    There is an efficient parameterized approximation scheme to find the $TH$-representation of $G$ when parameterized with the cluster vertex deletion number.
\end{theorem}
\begin{proof}
    Let $\epsilon>0$. Let $X$ be a cluster vertex deletion set of a graph $G=(V,E)$ with $|X|= k$. When $k=0$, we are done by Theorem~\ref{dimtgx}. So assume $k\geq 1$. Let $G_1=(V,E_1)$ be such that $E_1=E\cup \{uv\colon u\in V\setminus X, v\in V\setminus u\}$, $G_2=(V,E_2)$ such that $E_2=E\cup\{uv\colon u\in  X, v\in V\setminus u\}$, and $G_3=(V,E_3)$ be such that $E_3=\{uv\in E\colon u\in  X, v\in V\setminus X\}\cup \{uv\colon u,v\in X\t{ and }u\neq v\}\cup \{uv\colon u,v\in V\setminus X\t{ and }u\neq v\}$. Then clearly, $G=G_1\cap G_2\cap G_3$. 
    
    We have $\dimt(G_1)=\dimt(G- (V\setminus X))\leq k$ because it is the graph obtained from $G$ by making all vertices in $V\setminus X$ into universal vertices. A $TH$-representation for $G_1$ can be obtained in $2^{O(k^2\log k)}$ time (see the discussion in Section~\ref{prelim}).
    
    $G_3$ is a co-bipartite graph with one part of cardinality at most $k$. Therefore, the vertex cover of $\overline{G_3}$ can be at most $k$ and it follows from Theorem~\ref{thpoly} that $\dimt(G_3)\leq k$ and its $TH$-representation can be computed in polynomial time. 

     If $G-X$ contains at least 2 cliques, Theorem~\ref{dimtgx} tells us that the threshold cover number of $\overline{G_2}$ is equal to the order of the second largest clique of $G-X$ because $\overline{G_2}$ is a union of isolated vertices with a complete multipartite graph with the second largest part having cardinality $p$. In fact, $\dimt(G_2)=\dimt(G_2-X)$. In addition, a $TH$-representation for $G_2$ can be computed in polynomial time. Since $G$ contains $G-X$ as an induced subgraph, we must have $\dimt(G)\geq \dimt(G-X)=\dimt(G_2-X)=\dimt(G_2)$.  

    We have three cases.
    \begin{enumerate}
        \item The first case is when $G-X$ has at least 2 cliques, the order of the second largest clique in $G-X$ is $p$ and $2k\leq \epsilon p$. By the above discussion, we have $\dimt(G)\leq \dimt(G_1)+\dimt(G_2)+\dimt(G_3)\leq 2k+p\leq(1+\epsilon)p\leq (1+\epsilon)\dimt(G_2)\leq (1+\epsilon)\dimt(G)$. In this case, the algorithm returns the union of the $TH$-representations of $G_1$, $G_2$ and $G_3$, to get a $(1+\epsilon)$-factor approximation for $\dimt(G)$. This takes $2^{O(k^2\log k )}+O(n^4)$ time which is a consequence of the time taken to compute $\dimt(G_1)$ along with the time taken to find a 3-approximate cluster vertex deletion set.   
        \item The second case has 2 subcases:
        \begin{romanenumerate}
            \item  Consider the case when $G-X$ has at least 2 cliques and the order of the second largest clique in $G-X$ is $p$ satisfying $p<\frac{2k}{\epsilon}$. In this case, we have $\dimt(G)< 2k(1+\frac{1}{\epsilon})$. Observe that each clique in $G-X$ is of cardinality at most $\frac{2k}{\epsilon}$ except, possibly, for the largest clique $C_{max}$. By Lemma~\ref{cliqmax}, there is a subgraph $G'$ of $G$ such that $\dimt(G')=\dimt(G)$ and the largest clique in $G'-X$ has order at most $2^k(2k(1+\frac{1}{\epsilon})+1)\leq2^k(2k+1)(1+\frac{1}{\epsilon})$. By Lemma~\ref{struct1}, there is a graph $G''$ with $\dimt(G'')=\dimt(G')=\dimt(G)$ of order at most $\frac{4k}{\epsilon}\cdot(2^k+{\frac{2k}{\epsilon}})^{\frac{2k}{\epsilon}} (2^{\frac{4k^2}{\epsilon}(1+\frac{1}{\epsilon})}(k+1)^{\frac{4k^2}{\epsilon}(1+\frac{1}{\epsilon})}+1)+2^k(2k+1)(1+\frac{1}{\epsilon})+k$.
        \item Finally we have the case when $G-X$ has only one clique. In this case, we have $G=G_1\cap G_3$ and therefore, $\dimt(G)\leq 2k$. By Lemma~\ref{cliqmax}, there is a subgraph $G'$ of $G$ such that $\dimt(G')=\dimt(G)$ and the largest clique in $G'-X$ has order at most $2^k(2k+1)$. Therefore, the order of $G'$ is at most $2^k(2k+1)+k$. Let $G''=G'$.
        \end{romanenumerate}
        By the discussion in Section~\ref{prelim}, we can obtain a $TH$-representation for $G''$ by a bruteforce enumeration. And from $G''$ we can obtain an optimal $TH$-representation of $G'$ and subsequently an optimal $TH$-representation for $G$ in polynomial time.
    \end{enumerate} 
This finishes the proof.         
    \end{proof}
From the proof of Theorem~\ref{EPAS}, it is also clear that the problem is FPT with two parameters, the cluster vertex deletion number and the solution size. We record this in the following theorem. 
 \begin{theorem}\label{FPT}
         There is an algorithm that, given a graph \(G\) with $\dimt(G)=t$, constructs a $TH$-representation of \(G\) of size $t$ in time $f(k,t)\,|V(G)|^{O(1)},$ where \(k\) is the cluster vertex deletion number of \(G\), and \(f\) is a computable function.
     \end{theorem}

 \bibliography{lipics-v2021-sample-article}

@article{chacko2021representing,
  title={Representing Graphs as the Intersection of Cographs and Threshold Graphs}, volume={28},  DOI={10.37236/9110}, number={3}, journal={The Electronic Journal of Combinatorics}, author={Chacko, Daphna and Francis, Mathew C.}, year={2021}, month={Jul.}, pages={P3.11} 
}

@article{Para,
author={Adiga, Abhijin
and Chitnis, Rajesh
and Saurabh, Saket},
title={Parameterized Algorithms for Boxicity},
journal={Algorithms and Computation},
year={2010},
}

@article{Kratochvil, title={Intersection dimensions of graph classes}, 
author={Kratochvíl, Jan and Tuza, Zsolt},
journal={Graphs and Combinatorics}, 
volume={10}, 
number={2}, 
pages={159--168},
year={1994}, 
publisher={Springer} }

@article{francis2025some,
  title={Some Bounds on the Threshold Dimension of Graphs},
  author={Francis, Mathew C and Majumder, Atrayee and Mathew, Rogers},
  volume={32},
  journal={The Electronic Journal of Combinatorics},
  pages={P1--33},
  year={2025},
  DOI={10.37236/12331},
  month={Feb.}, pages={P1.33}
}

@book{mahadev1995threshold,
  title={Threshold graphs and related topics},
  author={Mahadev, Nadimpalli VR and Peled, Uri N},
  volume={56},
  year={1995},
  publisher={Elsevier}
}

@article{fellows2012well,
  title={Well quasi orders in subclasses of bounded treewidth graphs and their algorithmic applications},
  author={Fellows, Michael R and Hermelin, Danny and Rosamond, Frances A},
  journal={Algorithmica},
  volume={64},
  number={1},
  pages={3--18},
  year={2012},
  publisher={Springer}
}

@article{chvatal1973set, 
title={Set-packing and threshold graphs}, 
author={Chv{\'a}tal, V{\'a}cl{\'a}v and Peter L. Hammer}, 
journal={Res. Rep., Comput. Sci. Dept., Univ. Waterloo}, 
year={1973} }

@article{cozzens1989dimensional,
  title={On dimensional properties of graphs},
  author={Cozzens, Margaret Barry and Roberts, Fred S},
  journal={Graphs and Combinatorics},
  volume={5},
  number={1},
  pages={29--46},
  year={1989},
  publisher={Springer}
}

@article{doi:10.1137/0605055,
author = {Cozzens, Margaret B. and Leibowitz, Rochelle},
title = {Threshold Dimension of Graphs},
journal = {SIAM Journal on Algebraic Discrete Methods},
volume = {5},
number = {4},
pages = {579-595},
year = {1984},
doi = {10.1137/0605055},

URL = { 
    
        https://doi.org/10.1137/0605055
    
    

}
}

@article{bruhn2016structural,
  title={Structural parameterizations for boxicity},
  author={Bruhn, Henning and Chopin, Morgan and Joos, Felix and Schaudt, Oliver},
  journal={Algorithmica},
  volume={74},
  number={4},
  pages={1453--1472},
  year={2016},
  publisher={Springer}
}

@inproceedings{raschle1995recognition,
  title={Recognition of graphs with threshold dimension two},
  author={Raschle, Thomas and Simon, Klaus},
  booktitle={Proceedings of the Twenty-seventh Annual ACM Symposium on Theory of computing},
  pages={650--661},
  year={1995}
}

@techreport{Ma,
    author = {Tze-Heng Ma},
    title = { On the threshold dimension 2 graphs},
    institution = {Institute of Information 
Science, Nankang, 
Taipei, 
Taiwan},
    year = {September 1993}
}

@article{sterbini1998n3,
  title={An {O}(n3) time algorithm for recognizing threshold dimension 2 graphs},
  author={Sterbini, Andrea and Raschle, Thomas},
  journal={Information Processing Letters},
  volume={67},
  number={5},
  pages={255--259},
  year={1998},
  publisher={Elsevier}
}

@article{Yannakakis,
author = {Yannakakis, Mihalis},
title = {The Complexity of the Partial Order Dimension Problem},
journal = {SIAM Journal on Algebraic Discrete Methods},
volume = {3},
number = {3},
pages = {351--358},
year = {1982},
doi = {10.1137/0603036}
}
\newpage
 \appendix
\section{Proofs of some Claims and Lemmas}
Since Lemma~\ref{integth} follows from Corollary~\ref{integth1}, we will be proving Corollary~\ref{integth1}.
\begin{proof}[Proof of Corollary~\ref{integth1}]

    By definition~\ref{orderdef}, there is an order $v_1, v_2, \ldots ,v_n$ of the vertices of $G$ such that a vertex at a particular position is either adjacent to all the vertices appearing before it in the order or it is not adjacent to any of the vertices appearing before it in the order. For $i\in [n]$, define $g(v_i)=n-i$ and let $w(v_i)=2n-g(v_i)$ if $v_i$ is adjacent to all the vertices appearing before it in the order and $w(v_i)=g(v_i)$ if $v_i$ is not adjacent to any of the vertices appearing before it in the order. Let $c(v_i)=0$ if $w(v_i)<n$ and $c(v_i)=1$ if $w(v_i)\geq n$. It is clear that $\Gamma(2n,w)$ is the threshold representation of $G$.  
    
    Suppose that the elements of $g(V)$ are ordered as $x_1< x_2<\ldots <x_n$. In this order, take the subsequence corresponding to the vertices in $X$ as $y_1<y_2<\ldots <y_k$. Define $y_{k+1}=1+ x_n$ and $y_0= x_1-1$. Define the function $z$ on the elements in $g(V)$ as follows: $$z(x)=\begin{cases}
        j & \t{ if } x=y_j \t{ for some }j\in [k]\\
        \frac{x-y_{j-1}}{y_{j}-y_{j-1}}+j -1& \t{ if } y_{j-1}<x< y_{j} \t{ for some }j\in [k+1]
    \end{cases}.$$ Note that $z$ preserves the ordering of $g(V)$ and takes value less than $k+1$, i.e., we have $z(x_1)< z(x_2)<\ldots <z(x_n)<k+1$.  We define the function $w^*$ on the vertices of $G$ as follows: 
    $$w^*(v)=\begin{cases}
        z(g(v)) & \t{ if }w(v)<n\\
        2k+3-z(g(v)) & \t{ if }  w(v)\geq n
    \end{cases}
    $$ 
    Observe that if $w(v)< n$, we have $w^*(v)< k+1$ and if $w(v)\geq n$, we have $w^*(v)> k+2$. 
    Also observe that for every vertex $v$ in $V$, we have $g(w^*,v)=z(g(v))$. Therefore, it is clear that $g(w^*,\cdot)$ satisfies requirements (i), (ii) and (iii). We will show that $\Gamma(2k+3,w^*)$ is a valid threshold representation of $G$.

    Suppose for distinct vertices $u,v\in V$ we have $uv\notin E$. Then either $w(u)<n$ or $w(v)<n$. Without loss of generality, assume $w(u)<n$. If we also have $w(v)<n$ then $ w^*(u)+ w^*(v)=z(g(u))+z(g(v))<k+1+k+1= 2k+2<2k+3$ and we are done. Therefore, assume that $w(v)\geq n$. Since $w(u) +w(v)<2n$, we have $2n-w(v)>w(u)$. It follows that $g(v)>g(u)$ and, consequently, $z(g(v))>z(g(u))$. Hence, $w^*(u)+ w^*(v)=z(g(u))+2k+3-z(g(v))<2k+3$. 

    On the other hand, if for distinct vertices $u,v\in V$ we have $uv\in E$, then $w(u)+w(v)\geq 2n$. This means that either $w(u)\geq n$ or $w(v)\geq n$. Without loss of generality, assume that $w(u)\geq n$. If we also have $w(v)\geq n$, then $w^*(u)+ w^*(v)> k+2+k+2=2k+4>2k+3$ as discussed above. Therefore, assume that $w(v)<n$. Observe that $2n-w(u)\leq w(v)$. Consequently, $g(u)\leq g(v)$. Since $g$ takes distinct values, we must have $g(u)<g(v)$. It follows that $z(g(u))<z(g(v))$. Therefore, we have $w^*(u)+ w^*(v)=2k+3-z(g(u))+z(g(v))> 2k+3$.  
\end{proof}
\begin{proof}[Proof of Claim~\ref{claimt} case 2(iii)]
    Assume that $c(w_i,u)=1$ and $c(w_i,f(q(v)))=0$. This means that 
    \begin{equation}\label{eq6}
        c(w_i,f(u))=1 \t{ and }c(w_i,v)=0
    \end{equation}
    because $u$ and $f(u)$ have the same configuration in $\Gamma(2k+3,w_i)$ and $v,q(v)$ and $f(q(v))$ have the same configuration in $\Gamma(2k+3,w_i)$. 
    By Claim~\ref{conf} and expression~(\ref{eq6}), this means that $g(w_i,f(q(v)))\geq g(w_i,f(u))$. 
    
    Since $c(w_i,u)=1$ and $c(w_i,f(q(v)))=0$, we have \begin{equation}\label{eq7}
    g(w_i,u)=2k+3-w_i(u)>w_i(f(q(v)))=g(w_i,f(q(v)))\geq g(w_i,f(u)).
    \end{equation} 
    Since  $u$ and $f(u)$ have the same configuration in $\Gamma(2k+3,w_i)$, it follows from expression~(\ref{eq7}) that 
    \begin{equation}\label{eq8}
        \lceil g(w_i,u)\rceil =\lceil g(w_i,f(u))\rceil =\lceil g(w_i,f(q(v)))\rceil=\lceil g(w_i,v)\rceil=\lceil g(w_i,q(v))\rceil.
    \end{equation} By Claim~\ref{conf} and expression~(\ref{eq8}), we have \begin{equation}\label{eq9}
        g(w_i,q(v))= g(w_i,u)>g(w_i,f(q(v))).
    \end{equation} Therefore, by Equation~(\ref{eq}), we have $g(w_i,v)=g(w_i,q(v))-\delta_i $. By the definition of $g(w_i,\cdot)$ and expression~(\ref{eq9}), we have $w_i(v)=g(w_i,v)=g(w_i,q(v))-\delta_i= g(w_i,u)-\delta_i=2k+3-w_i(u)-\delta_i$. Therefore, $w_i(u)+w_i(v)=w_i(u)+2k+3-w_i(u)-\delta_i=2k+3-\delta_i<2k+3$.
\end{proof}
\begin{claimproof}[Proof of Claim~\ref{poly}]
     Let $|\mathcal{C}|=c$. We first describe a polynomial time procedure using which we can determine all the equivalence classes of $\mathcal{C}$ and also determine for every equivalent pair $C,C'\in \mathcal{C}$ a labelling to its vertices such that if $v\in C$ and $v'\in C'$ are labeled the same then $N_G(v)\cap X=N_G(v')\cap X$. Suppose the number of equivalence classes of $\mathcal{C}$ is $r$.  
\begin{itemize}
    \item Initialize $\mathcal{C}=$ cliques in $G'-X$ (can be done efficiently by observing the neighborhoods). Initialize all the vertices of $G'-X$ with the label NULL. Initialize $i=0$.
    \item while $\mathcal{C}\neq \emptyset$, do:
    \begin{itemize}
    \item $i=i+1$.
    \item Pick a clique $C_i\in \mathcal{C}$ arbitrarily and assign an arbitrary labeling to its vertices from the integers in $[l]$. Initialize $\mathcal{U}_{C_i}=\{C_i\}$.
    \item For every clique $C'\in \mathcal{C}\setminus \{C_i\}$ that has the same order as $C_i$, do the following: 
    \begin{itemize}
        \item Initialize a flag variable with 1.
        \item For every vertex $v\in C_i$ 
        \begin{itemize}
            \item Find a vertex $v'\in C'$ (not found in a previous step) such that $N_G(v)\cap X=N_G(v')\cap X$. If such a vertex is found, then label $v'$ with the same label as $v$ (at most one vertex is labelled this way). 
            \item If no such vertex is found then set flag as 0 and stop. 
        \end{itemize}
        \item If the flag is 0, then reassign all the labels of the vertices of $C'$ with NULL. 
        \item If the flag is 1, then do $\mathcal{U}_{C_i}=\mathcal{U}_{C_i}\cup \{C'\}$.
    \end{itemize}
    \item $\mathcal{C}=\mathcal{C}\setminus \mathcal{U}_{C_i}$.
\end{itemize}
    \item Return $\{\mathcal{U}_{C_1},\mathcal{U}_{C_2},\ldots , \mathcal{U}_{C_r}\}$.
    \end{itemize}
     Now we construct a $c\times tl$ matrix as follows: the rows of the matrix are labeled with the cliques in $\mathcal{C}$ and the columns are labeled with $t$ many copies of the labels that the vertices of the cliques received by the above procedure. Let $C\in \mathcal{C}$ and $v\in C$ be a vertex that received the label $j$ in the above procedure. Then for $i\in [t]$, the $(C,j+(i-1)l)$-th cell is filled with the tuple $(c(w_i,v),\lceil g(w_i,v)\rceil)$. All the other cells are filled with a filler value $(*,*)$. This matrix has size at most $n^3$ and therefore, takes polynomial time to fill. We can now identify two equivalent cliques $C_1$ and $C_2$ that are in the same configuration and in $\langle\tilde{C}\rangle$ by going down the rows of this matrix and checking if two rows are equal. This again takes polynomial time. Once we have identified $C_1$ and $C_2$ we can assign the weights of $\tilde{C}$ according to Equation~(\ref{eq}) and we are done.
\end{claimproof}

\end{document}